\documentclass{article}
\usepackage{fullpage}
\usepackage{amsmath, amssymb}
 \usepackage{amsthm}
 \usepackage{amstext}
\usepackage{verbatim}
\usepackage{graphicx, epsfig}
\usepackage{algorithm, algorithmic}

\newcommand{\OPT}{\operatorname{OPT}}

\newcommand{\Scal}{\mathcal{S}}

\newcommand{\T}{\mathcal{T}}

\newcommand{\diam}{\operatorname{diam}}

\newcommand{\C}{\mathcal{C}}

\newcommand{\class}{\operatorname{class}}

\newcommand{\ALG}{\operatorname{ALG}}

\newcommand{\vF}{\hat{F}}

\newcommand{\vassign}{\hat{\sigma}}

\newcommand{\Lead}{\operatorname{Leader}}
\newcommand{\Leads}{\operatorname{Leaders}}
\newcommand{\OFL}{\operatorname{OFL-ALG}}

\newcommand{\Rmax}{R_{\operatorname{max}}}
\newcommand{\SOL}{\operatorname{SOL}}
\newcommand{\D}{\mathcal{D}}
\newcommand{\OPTFL}{\OPT_{\operatorname{FL}}}
\newcommand{\OPTROB}{\OPT_{\operatorname{ROB}}}

\newcommand{\OPTSF}{\OPT_{\operatorname{SF}}}
\newcommand{\SF}{\operatorname{SF-ALG}}

\newtheorem{theorem}{Theorem}[section]
\newtheorem{lemma}[theorem]{Lemma}
\newtheorem{corollary}[theorem]{Corollary}
\newtheorem{claim}{Claim}[section]
\newtheorem{definition}{Definition}
\begin{document}

\title{\Large Online Network Design Algorithms 
via
 Hierarchical Decompositions\thanks{This work was done while the
   author was visiting Microsoft Research and the University of
   Washington. Supported in part by NSF Award CCF-1320854.}}
\author{Seeun Umboh\thanks{University of
  Wisconsin -- Madison. \tt{seeun@cs.wisc.edu}.}}
\date{}

\maketitle 

\begin{abstract}
We develop a new approach for online network design and obtain
improved competitive ratios for several problems. Our approach gives
natural deterministic algorithms and simple analyses. At the heart of
our work is a novel application of embeddings into hierarchically
well-separated trees (HSTs) to the analysis of online network design
algorithms --- we charge the cost of the algorithm to the cost of the
optimal solution on any HST embedding of the terminals. This analysis
technique is widely applicable to many problems and gives a unified
framework for online network design.

In a sense, our work brings together two of the main approaches to
online network design. The first uses greedy-like algorithms and
analyzes them using dual-fitting. The second uses tree embeddings ---
embed the entire graph into a tree at the beginning and then
solve the problem on the tree --- and results in randomized $O(\log
n)$-competitive algorithms, where $n$ is the total number of vertices
in the graph. Our approach uses deterministic greedy-like algorithms
but analyzes them via HST embeddings of the terminals. Our proofs are
simpler as we do not need to carefully construct dual solutions and we
get $O(\log k)$ competitive ratios, where $k$ is the number of
terminals.

In this paper, we apply our approach to obtain deterministic $O(\log
k)$-competitive online algorithms for the following problems.
\begin{enumerate}
\item Steiner network with edge duplication. Previously, only a
  randomized $O(\log n)$-competitive algorithm was known.
\item Rent-or-buy. Previously, only deterministic $O(\log^2
  k)$-competitive and randomized $O(\log k)$-competitive algorithms by
  Awerbuch, Azar and Bartal (\emph{Theoretical Computer Science} 2004)
  were known.
\item Connected facility location. Previously, only a
  randomized $O(\log^2 k)$-competitive algorithm of San Felice,
  Williamson and Lee (\emph{LATIN} 2014) was known.
\item Prize-collecting Steiner forest. We match the competitive ratio
  first achieved by Qian and Williamson (\emph{ICALP} 2011) and give a
  simpler analysis.
\end{enumerate}

Our competitive ratios are optimal up to constant factors as these
problems capture the online Steiner tree problem which has a lower
bound of $\Omega(\log k)$.

\end{abstract}

\section{Introduction}
\label{sec:intro}
We study network design problems in the online model. In a network
design problem, we are given a graph with edge costs and connectivity
requirements. The goal is to find a minimum-cost subgraph satisfying
the requirements. In the online model, requests in the form of
terminals and connectivity requirements arrive one-by-one, and the
goal is to maintain a minimum-cost subgraph satisfying all arrived
requirements. The challenge in the online model is that decisions are
irrevocable: once an edge is added to the subgraph, it cannot be
removed later on. For instance, the classic Steiner tree problem asks
for a minimum-cost subgraph connecting a given set of terminals. In
the online Steiner tree problem~\cite{doi:10.1137/0404033}, terminals
arrive one-by-one and the algorithm maintains a subgraph connecting
all terminals so far. In the following, we use $n$ to denote the
number of vertices in the underlying graph, and $k$ the number of
terminals.



Our main contribution is a novel application of embeddings into
hierarchically well-separated trees (HSTs) \cite{Bartal:1996hg,
  Bartal:1998, Fakcharoenphol:2004cg} as a tool for analyzing online
network design algorithms.
In this paper, we apply this technique to several problems and obtain
natural algorithms that improve upon previous work. Our approach also
gives a simple unified analysis across these different problems.
%
%

At a high level, our work brings together two disparate lines of work
in online network design.  The first designs greedy-like algorithms
and analyzes them via dual-fitting. The idea is to grow a dual
solution in tandem with the algorithm such that at any step, the dual
accounts for the cost of the algorithm, mirroring the primal-dual
schema in the offline setting. Then one shows that the dual is
feasible after appropriate scaling. Berman and
Coulston~\cite{Berman:1997:OAS:258533.258618} used this approach to
give a $O(\log k)$-competitive algorithm for online Steiner forest
which was later extended by Qian and Williamson~\cite{Qian:2011wg} to
the more general online constrained forest and prize-collecting
Steiner forest problems. 
 
The other line of work is based on tree embeddings.
Awerbuch and Azar~\cite{646143} observed that for many online network
design problems, one can first probabilistically embed the entire
input graph into a tree \cite{Bartal:1996hg, Bartal:1998,
  Fakcharoenphol:2004cg} before the requests arrive and then
(essentially) solve the online problem on the tree. In the analysis,
one shows that for any choice of tree embedding, the cost of the
resulting solution can be charged to the cost of the optimal solution
on the tree. For problems that are easily solved on trees, this
results in a randomized algorithm whose competitive ratio is, in
expectation over the probabilistic embedding, at most $O(\log n)$.

We combine these two approaches in the following way: we use
greedy-like algorithms but analyze them using HST embeddings. In
particular, we develop a charging scheme showing that for \emph{any}
expanding embedding of the terminals into a HST $T$, the cost of the
algorithm is at most a constant times the cost of the optimal solution
on the tree $\OPT(T)$. We emphasize that the embedding only exists in
the analysis, not in the algorithm. This shows that the competitive
ratio of the algorithm is bounded by
\begin{equation}
\label{eq:cr}
O(1)\cdot \min_{T \in \T}\frac{\OPT(T)}{\OPT},
\end{equation}
where $\T$ is the set of expanding HST embeddings of the terminals and
$\OPT$ is the cost of the optimal solution on the input graph. The
bound \eqref{eq:cr} is at most the expected distortion of the
embedding of the terminals into HSTs which is $O(\log k)$
\cite{Bartal:1996hg, Bartal:1998, Fakcharoenphol:2004cg}. However,
since \eqref{eq:cr} optimizes with respect to a particular subgraph,
namely $\OPT$, the bound can be much better depending on the
instance. Unlike previous work using tree embeddings, we get better
competitive ratios (deterministic $O(\log k)$ instead of randomized
$O(\log n)$). 
Furthermore, our algorithms are natural greedy-like algorithms that do
not need to know the entire underlying graph upfront.

Our work highlights an interesting connection between HST embeddings
and dual-fitting. We can interpret our charging scheme in terms of
dual-fitting --- given any expanding HST embedding of the terminals,
we build a dual solution that is feasible for the HST and charge
against it. These dual solutions can be highly infeasible for the
original graph, but averaging over the probabilistic embeddings of
\cite{Fakcharoenphol:2004cg} gives a dual solution that is feasible
after scaling by the embedding distortion. However, our approach
differs from dual-fitting in the usual sense --- we charge against
\emph{multiple} dual solutions simultaneously (one per embedding), and
the dual solutions are not built with respect to the original metric
but with respect to the HST embeddings of the terminals. In a sense,
our work uses HST embeddings as a black-box to generate good dual
solutions. Compared to the usual dual-fitting approach, we hide the
complexity of the dual construction within the HST embedding, allowing
us to give simpler algorithm descriptions since the algorithm no
longer needs to take the dual into account. Furthermore, we get a more
streamlined analysis and a unified approach to different problems. In
some cases, we also get a tighter upper bound (e.g. see footnote 4 in Section
\ref{sec:SN}).

\subsection{Our Results}
In this paper, we illustrate our technique on a variety of online network design problems. In all of these problems, there is an underlying graph $G = (V,E)$ with edge lengths $d(u,v)$. We assume w.l.o.g. that $G$ is a complete graph, the edge lengths satisfy the triangle inequality, i.e. $(V,d)$ is a metric space, and that the minimum edge length is $1$. Initially, the algorithm does not know $G$; at each time step online, it only knows the submetric over the arrived terminals.



For a (multi-)graph $H$, we define its \emph{cost} $c(H)$ to be the sum of lengths of edges in $H$. We use $n$ to denote the number of vertices $|V|$ in the graph $G$, and $k$ the number of terminals that arrive online. We remark that our competitive ratios below are optimal up to constants as these problems capture the online Steiner tree problem which has a $\Omega(\log k)$ lower bound~\cite{doi:10.1137/0404033}.

\paragraph{Steiner problems.}
In the \emph{online Steiner tree} problem, the algorithm is given a root terminal $r$ at the beginning. Terminals $i$ arrive online one-by-one and the algorithm maintains a subgraph $H$ connecting terminals to the root. In the \emph{online Steiner forest problem}, terminal pairs $(s_i,t_i)$ arrive online one-by-one and the algorithm maintains a subgraph $H$ in which each terminal pair is connected. In the \emph{online Steiner network problem with edge duplication}, each $(s_i,t_i)$ pair comes with a requirement $R_i$, and the algorithm maintains a multigraph $H$ which contains $R_i$ edge-disjoint $(s_i,t_i)$-paths. 
Note that allowing $H$ to be a multigraph means that the algorithm is allowed to buy multiple copies of an edge. For brevity, we will simply call this the \emph{online Steiner network problem}. The goal in these problems is to maintain a minimum-cost (multi-)graph $H$ satisfying the respective requirements subject to the constraint that once an edge is added to $H$, it cannot be removed later on.

Previously, deterministic $O(\log k)$-competitive algorithms were known for the online Steiner tree and Steiner forest problems~\cite{doi:10.1137/0404033,Berman:1997:OAS:258533.258618}. However, the best algorithm (as far as we know) for the online Steiner network problem is to use tree embeddings and yields a randomized $O(\log n)$ competitive ratio. Our first result closes the gap between the known competitive ratios of these problems.

\begin{theorem}
  \label{thm:proper}
  There is a deterministic $O(\log k)$-competitive algorithm for the
  online Steiner network problem (with edge duplication). 
\end{theorem}

We remark that our approach also gives a simpler analysis of the Berman-Coulston online Steiner forest algorithm~\cite{Berman:1997:OAS:258533.258618}.

\paragraph{Rent-or-buy.} The rent-or-buy problem generalizes the Steiner forest problem. The algorithm is allowed to either \emph{rent} or \emph{buy} edges in order to satisfy a request. Buying an edge costs $M$ times more than renting, but once an edge is bought, it can be used for free in the future. On the other hand, a rented edge can only be used once; future terminals have to either re-rent it or buy it in order to use it.


More formally, in the \emph{online rent-or-buy problem}, the algorithm
is also given a parameter $M \geq 0$. The algorithm maintains a
subgraph $H$ of bought edges; when a terminal pair $(s_i, t_i)$
arrives, the algorithm buys zero or more edges and rents edges $Q_i$
such that $H \cup Q_i$ connects $s_i$ and $t_i$. Both rent and buy
decisions are irrevocable --- edges cannot be removed from $H$ later on and $Q_i$ is fixed after the $i$-th step.
The total cost of the algorithm is $Mc(H) + \sum_i c(Q_i)$. In the
\emph{online single-source rent-or-buy problem}, the algorithm is also
given a root terminal $r$ in advance; terminals $i$ arrive online and
$i$ has to be connected to $r$ in the subgraph $H \cup Q_i$.

Previously, Awerbuch, Azar and Bartal~\cite{Awerbuch:2004tn}\footnote{They called it the \emph{network connectivity leasing problem}.} 
%
gave a randomized $O(\log k)$-competitive algorithm as well as a deterministic $O(\log^2 k)$-competitive algorithm and posed the existence of a deterministic $O(\log k)$-competitive algorithm as an open problem. Our second result resolves this positively.

\begin{theorem}
  \label{thm:MROB}
  There is a deterministic $O(\log k)$-competitive algorithm for the
  online rent-or-buy problem.
\end{theorem}

\paragraph{Connected facility location.} In the \emph{online connected
  facility location problem}, we call terminals \emph{clients}. At the
beginning, the algorithm is given a parameter $M \geq 0$, a set of
\emph{facilities} $F \subseteq V$ and facility opening costs $f_x$ for
each facility $x \in F$. There is also a designated root facility $r
\in F$ with zero opening cost. The algorithm maintains a set of
\emph{open} facilities $F' \subseteq F$ and a Steiner tree $H$
connecting $F'$ and $r$. At any online step, the algorithm knows the
submetric over the arrived clients and the facilities. When a client $i$
arrives, the algorithm may open a new facility, and then assigns $i$ to
some open facility $\sigma(i) \in F'$. All decisions are irrevocable
--- edges cannot be removed from $H$ later on, clients cannot be
reassigned and opened facilities cannot be closed. The total cost of
the algorithm is $\sum_{z \in F'} f_z + Mc(H) + \sum_i d(i,
\sigma(i))$. When $M = 0$ (i.e. the facilities need not be connected),
this is called the \emph{online
  facility
  location
  problem} \cite{Meyerson:2001io}. For these problems, we use $k$ to denote the number of
clients.

This problem was recently proposed by San Felice, Williamson and Lee~\cite{felice} and they gave a randomized $O(\log^2 k)$-competitive algorithm
\footnote{At the time of submission, the author heard from
  San Felice (personal communication) that he independently obtained a
  randomized $O(\log k)$-competitive algorithm recently.}. Our third result improves on their work.

\begin{theorem}
  \label{thm:CFL}
  There is a deterministic $O(\log k)$-competitive algorithm for the
  online connected facility location problem.
\end{theorem}

\paragraph{Prize-collecting Steiner forest.} 
In the \emph{online prize-collecting Steiner forest} problem, each terminal pair $(s_i,t_i)$ arrives with a penalty $\pi_i$ and the algorithm either pays the penalty or connects the pair. 
The total cost of the algorithm is $c(H) + \sum_{i \in P} \pi_i$ where $H$ is the subgraph maintained by the algorithm and $i \in P$ if the algorithm paid the penalty for $(s_i, t_i)$. Note that penalties are irrevocable: once the algorithm decides to pay the penalty $\pi_i$, the penalty does not get removed from the algorithm's cost even if later on $H$ ends up connecting $s_i$ and $t_i$.

For this problem, we obtain the same deterministic $O(\log k)$ competitive ratio first achieved by Qian and Williamson~\cite{Qian:2011wg} but our charging scheme is simpler than their dual-fitting analysis.


\begin{theorem}[\cite{Qian:2011wg}]
  \label{thm:PCSF}
  There is a deterministic $O(\log k)$-competitive algorithm for the
  online prize-collecting Steiner forest problem.
\end{theorem}

\paragraph{Our Techniques.}

For each problem, we design a natural greedy algorithm and show that the cost of the algorithm is at most a constant times the cost of the optimal solution on any HST embedding of the terminals. The key structural property of a HST embedding $T$ is that each edge of $T$ defines a cut $C \subseteq X$ whose diameter is proportional to the length of the edge (see Definition \ref{def:HST}), so we can express $\OPT(T)$ in terms of contributions from bounded-diameter cuts. In the analysis, we charge the cost of the algorithm to these cuts 
by dividing the cost of the algorithm among terminals and charging the cost share of each terminal to a bounded-diameter cut of $T$ containing it. We then argue that since the algorithm is greedy and the cuts have bounded diameter, each cut receives a charge of at most a constant times its contribution to $\OPT(T)$. Thus, the cost of the algorithm is at most $O(1)\OPT(T)$ for any HST embedding $T$ of the terminals. Since \cite{Fakcharoenphol:2004cg} gives a probabilistic HST embedding of the terminals with expected distortion $O(\log k)$, there exists a HST embedding $T^*$ such that $\OPT(T^*) \leq O(\log k)\OPT$. This gives us the desired $O(\log k)$ bound on the competitive ratio of the algorithm. Note that we do not need any particular property of the embedding of \cite{Fakcharoenphol:2004cg}. It is only used to prove the existence of $T^*$.

\subsection{Related Work}
There is a long history of research on network design in the offline setting. Agrawal, Klein and Ravi~\cite{Agrawal:1995:TCA:204695.204699} gave a
primal-dual $2$-approximation algorithm for the Steiner forest
problem. This was later extended by Goemans and
Williamson~\cite{doi:10.1137/S0097539793242618} to give a
$2$-approximation for the more general constrained forest and
the prize-collecting Steiner tree problems. 
In a breakthrough result, Jain~\cite{jain2001factor} gave a
$2$-approximation algorithm using iterative LP rounding for the
Steiner network problem. We remark that Jain's result holds even when
edge duplication is not allowed. Swamy and Kumar~\cite{Swamy:2004tc}
gave a primal-dual $4.55$-approximation algorithm for the
single-source rent-or-buy problem and a $8.55$-approximation algorithm
for the connected facility location problem. Later, Gupta et al.
\cite{gupta2007approximation} improved the constants using an elegant
sample-and augment approach.

Online network design was first studied by Imase and Waxman~\cite{doi:10.1137/0404033} who showed that the greedy algorithm is $O(\log k)$-competitive for the online Steiner tree problem and gave a matching lower bound. Awerbuch, Azar and Bartal~\cite{Awerbuch:2004tn} showed that a greedy algorithm is $O(\log^2 k)$-competitive for the online Steiner forest problem; later, Berman and Coulston~\cite{Berman:1997:OAS:258533.258618} gave a $O(\log k)$-competitive algorithm. Recently, Qian and Williamson~\cite{Qian:2011wg} extended the approach of \cite{Berman:1997:OAS:258533.258618} to give an $O(\log k)$-competitive algorithm for the online constrained forest and prize-collecting Steiner forest problems. As far as we know, there is no previous work on the online Steiner network problem with edge duplication. When edge duplication is disallowed, unlike the offline setting, the problem becomes much harder. Gupta et al.~\cite{doi:10.1137/09076725X} showed a lower bound of $\Omega(\min\{k, \log n\})$ and an upper bound of $\tilde{O}(\Rmax \log^3n)$, where $\Rmax$ is the maximum requirement.

The online rent-or-buy problem was first considered by Awerbuch, Azar
and Bartal~\cite{Awerbuch:2004tn}, and they gave a deterministic
$O(\log^2 k)$-competitive algorithm and a randomized $O(\log
k)$-competitive algorithm. The online connected facility location
problem was recently proposed by San Felice et al~\cite{felice} and
they presented a randomized $O(\log^2 k)$-competitive algorithm based
on the offline connected facility location algorithm of Eisenbrand et
al.~\cite{eisenbrand2008approximating}. 
Surprisingly, 
Fotakis~\cite{fotakis2008competitive} showed that there is a
lower bound of $\Omega(\frac{\log
  k}{\log \log
  k})$ for the special case of the facility location problem ($M = 0$)
even when the underlying metric space is a HST.

For tree embeddings, Bartal showed that any graph can be
probabilistically embedded into HSTs with $O(\log^2 n)$ expected
distortion~\cite{Bartal:1996hg}
  and 
  subsequently improved this to $O(\log n \log \log n)$
  \cite{Bartal:1998}. Fakcharoenphol et
  al.~\cite{Fakcharoenphol:2004cg} achieved the tight $O(\log n)$
  bound.

\paragraph{Roadmap.}
We start with the necessary HST embedding definitions and results in Section \ref{sec:prelim}. Next, we warm up by illustrating our techniques to analyze the greedy algorithm for the Steiner tree problem in Section \ref{sec:steinertree}. In Section \ref{sec:steiner}, we apply our analysis framework to the Berman-Coulston algorithm for online Steiner forest algorithm \cite{Berman:1997:OAS:258533.258618} (yielding a simpler analysis) and show that this essentially allows us to reduce the online Steiner network instance to several online Steiner forest instances. 

We use the same high-level approach for the rent-or-buy, connected facility location and prize-collecting Steiner problems as they share very similar cost structures. The key ideas are illustrated using the single-source rent-or-buy problem in Section \ref{sec:SROB}. Then we extend these ideas to the multi-source setting in Section \ref{sec:MROB} and the connected facility location problem in Section \ref{sec:CFL}. 
The objective of the prize-collecting Steiner tree problem is closely related to that of the single-source rent-or-buy problem. In Section \ref{sec:PCST}, we show how a straightforward adaptation of the single-source rent-or-buy algorithm gives a prize-collecting Steiner tree algorithm. We omit a discussion of the prize-collecting Steiner forest algorithm as it is obtained via an identical adaptation of the multi-source rent-or-buy algorithm. 

We remark that the analyses for single-source rent-or-buy, connected facility location and prize-collecting Steiner tree rely on a property of the greedy online Steiner tree algorithm (Lemma \ref{lem:r-sep}), and the multi-source rent-or-buy algorithm depends on a property of the Berman-Coulston online Steiner forest algorithm (Lemma \ref{lem:SF:sep}).

\section{Basics of HST Embeddings}
\label{sec:prelim}
In this section, we present the necessary metric embedding definitions
and results. Let $(X,d)$ be a metric over a set of $k$ terminals with
the smallest distance $\min_{u,v\in X} d(u,v) = 1$. We are only
interested in \emph{expanding} embeddings, i.e. embeddings that do not
shrink distances.


\begin{definition}[Embeddings]
  A metric $(X',d')$ is an \emph{$\alpha$-distortion embedding} of
  $(X',d')$ with \emph{distortion} $\alpha \geq 1$ if $X' \supseteq X$
  and for all $u,v \in X$, we have $d(u,v) \leq d'(u,v) \leq \alpha
  d(u,v)$.
\end{definition}


In this paper, we will be concerned only with embeddings into
hierarchically separated trees~\cite{Bartal:1996hg}.

\begin{definition}
  [HST Embeddings] 
  \label{def:HST}
  A \emph{hierarchically separated tree (HST) embedding} $T$ of
  $(X,d)$ is a rooted tree with height $\lceil \log (\max_{u,v \in X}
  d(u,v)) \rceil$ and edge lengths that are powers of $2$ satisfying
  the following properties.
  \begin{enumerate}
  \item The leaves of $T$ are exactly the terminals $X$.
  \item The distance from any node to each of its children is the
    same.
  \item The edge lengths decrease by a factor of $2$ as one moves
    along a root-to-leaf path.
  \item For an edge $e$ in $T$, let $C_e \subseteq X$ be the subset of
    terminals that are separated by $e$ from the root. We require that
    if $e$ has length $2^{j-1}$, then $d(u,v) < 2^j$ for any $u,v \in
    C_e$.
  \end{enumerate}
%
\end{definition}

\noindent The last property will be crucial to our analyses. See
Figure \ref{fig:HST} for an example of a HST embedding.
\begin{figure}
  \centering
  \includegraphics[scale=0.5]{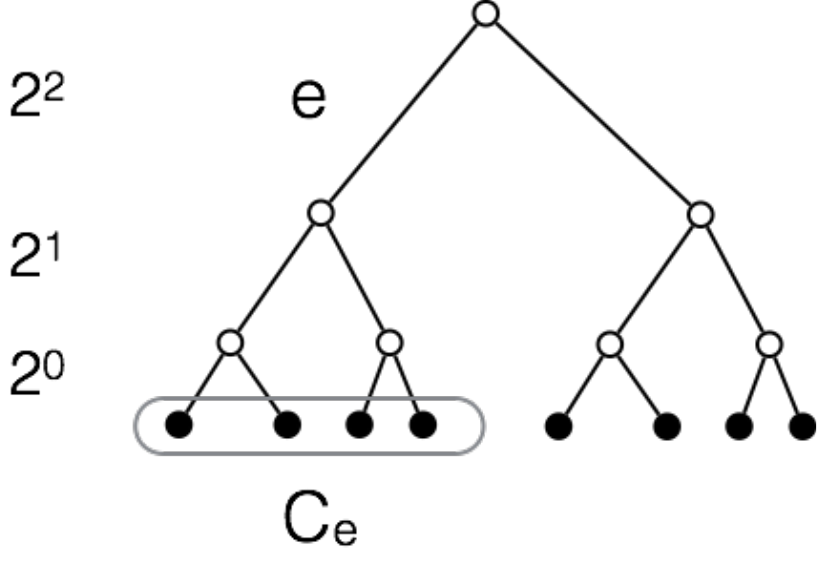}
  \caption{A binary HST embedding with $3$ levels. The terminals $X$
    are shown in solid black. The edge $e$ is a
    level-$3$ edge with length $2^2$ and $C_e$ is the set of terminals
    it separates from the root.}
  \label{fig:HST}
\end{figure}


A HST embedding $T$ of a metric space $(X,d)$ defines a hierarchical
decomposition of $(X,d)$ in the following way.  Call an edge $e$ of
length $2^{j-1}$ a \emph{level-$j$ edge} of $T$ and $C_e$ a
\emph{level-$j$ cut}. Denote by $E_j(T)$ the set of level-$j$ edges
and 
define $\C_j(T) = \{C_e : e \in E_j(T)\}$ to be the set of all
level-$j$ cuts. Now since the leaves of $T$ are exactly $X$, the 
level-$j$ cuts $\C_j(T)$ partitions $X$ into subsets of diameter less
than $2^j$. The family of partitions $\C(T) =
\{\C_j(T)\}_j$ is called the \emph{hierarchical decomposition} of
$(X,d)$ defined by $T$.  Note that there are $\lfloor \log (\max_{u,v
  \in X} d(u,v)) \rfloor$ levels and the level-$0$ cuts $\C_0(T)$ are
just the terminal singletons.

In the following, we will use the notation $\delta(C)$ to denote the
set of vertex pairs with exactly one endpoint in $C$, i.e. $\delta(C) =
\{(u,v) : |\{u,v\} \cap C| = 1\}$. Note that $(u,v) \in \delta(C_e)$
if and only if $e$ lies on the path between $u$ and $v$ in $T$.

We denote by $T(u,v)$ the distance between $u,v$ in
$T$. Fakcharoenphol et al.~\cite{Fakcharoenphol:2004cg} showed that
any metric can be embedded into HSTs with logarithmic expected
distortion.
\begin{theorem}[\cite{Fakcharoenphol:2004cg}]
  \label{thm:FRT}
  For any metric $(X,d)$, there exists a distribution $\D$ over HST
  embeddings $T$ such that $E_{T \sim \D}[T(u,v)] \leq O(\log
  k)d(u,v)$ for all $u,v \in X$.
\end{theorem}

The following corollary follows by standard arguments. We will apply
it to the Steiner network and rent-or-buy problems.

\begin{corollary}
  \label{cor:main}
  For any network design instance with terminals $X$ whose
  objective is a linear combination of edge lengths $d(u,v)$,
  there exists a distribution $\D$ over HST embeddings $T$ of $(X,d)$
  such that $E_{T \sim \D}[\OPT(T)] \leq O(\log k)\OPT$, where $\OPT$ (and $\OPT(T)$ resp.) is the cost of the optimal solution on
$(X,d)$ (and $T$ resp.).
\end{corollary}

Again, we emphasize that no particular property of the distribution
$\D$ is required. We only need the existence of a HST embedding $T^*$
such that $\OPT(T^*) \leq O(\log k)\OPT$.

\section{Warm up: Steiner Tree}
\label{sec:steinertree}

We warm up by applying our approach to the online Steiner tree problem 
and give an alternative proof that the greedy algorithm is $O(\log
k)$-competitive~\cite{doi:10.1137/0404033} in Section
\ref{sec:steinertree}. Recall the problem statement. The algorithm is
given a root terminal $r$ at the beginning. Terminals $i$ arrive
online and the algorithm maintains a subgraph $H$ connecting terminals
to the root. The greedy algorithm is very natural --- when terminal
$i$ arrives, it connects $i$ to the nearest previously-arrived
terminal. 
\paragraph{Analysis.}
Let $X \subseteq V$ be the set of terminals that arrived and $(X,d)$
be the submetric induced by $X$. The total cost of the greedy
algorithm is $\sum_i a_i$ where $a_i$ is the distance between
terminal $i$ and the nearest previously-arrived terminal.  Our goal is
to show that we can charge the cost of the greedy algorithm against
the cost of the optimal solution on any HST embedding of the
terminals.


Following the proof strategy outlined in the introduction, we first
define cost shares for each terminal such that the total cost share
accounts for the cost of the algorithm. We classify terminals
according to $a_i$ --- define $\class(i) = j$ if $a_i \in [2^j,
2^{j+1})$ and $Z_j \subseteq X$ to be the set of class-$j$
terminals. We define the cost share of each class-$j$ terminal $i \in
Z_j$ to be $2^{j+1}$, i.e. $i$'s cost share is $a_i$ rounded up to the
next power of $2$. Thus, the total cost share is at least the cost of
the algorithm.

\begin{lemma}
  \label{lem:r-steiner:share}
  $\sum_i a_i \leq \sum_j 2^{j+1}|Z_j|$.
\end{lemma}

Before we proceed, we show that class-$j$ terminals are at least
$2^j$-apart from each other.

\begin{lemma}
  \label{lem:r-sep}
  For any terminals $i, i' \in Z_j$, we have that $d(i,i') \geq
  2^j$.
\end{lemma}

\begin{proof}
  Suppose $i$ arrived later than $i'$. We have $a_i \leq d(i,i')$
  since $i$ could have connected to $i'$, and $a_i \geq 2^j$ by
  definition of $Z_j$. Thus, we get $d(i,i') \geq 2^j$.
\end{proof}

Next, we show that we can charge the cost shares against the cost of
the optimal solution on any HST embedding of the terminals. 


\begin{lemma}
  \label{lem:r-steiner:cost}
  $\sum_j2^{j+1}|Z_j| \leq O(1)\OPT(T)$ for any HST embedding $T$ of $(X,d)$.
\end{lemma}

\begin{proof}
  Let $T$ be a HST embedding of $(X,d)$.  
  We begin by expressing $\OPT(T)$ in terms of cuts from the
  hierarchical decomposition $\C(T)$.  Since the terminals are exactly
  the leaves of $T$, the unique feasible solution on $T$ is the entire
  tree. In particular, we have \[\OPT(T) = \sum_j 2^{j-1}|E_j(T)| =
  \sum_j
  2^{j-1}|\C_j(T)|.\] 

  For each terminal $i \in Z_j$, we charge $i$'s cost share to the
  level-$j$ cut containing it. The total charge received by a
  level-$j$ cut $C \in \C_j(T)$ is $2^{j+1}|Z_j\cap C|$ since we
  charge $2^{j+1}$ for each class-$j$ terminal in $C$. Thus we have
  \[\sum_j 2^{j+1}|Z_j| = \sum_j 2^{j+1}\sum_{C \in \C_j(T)} |Z_j \cap
  C|.\] 
  Each level-$j$ cut $C \in \C_j(T)$ has diameter less than $2^j$ so
  Lemma \ref{lem:r-sep} implies $|Z_j \cap C| \leq 1$. Therefore, we
  have $\sum_j 2^{j+1}|Z_j| \leq \sum_j 2^{j+1}|\C_j(T)| = 4\OPT(T)$.
\end{proof}

We are now ready to bound the competitive ratio of the greedy
algorithm. Lemmas \ref{lem:r-steiner:share} and
\ref{lem:r-steiner:cost} imply that the cost of the greedy algorithm
is at most $O(1)\OPT(T)$ for any HST embedding $T$ of
$(X,d)$. Furthermore, Corollary \ref{cor:main} implies that there
exists a HST embedding $T^*$ such that $\OPT(T^*) \leq O(\log
k)\OPT$. Thus, the greedy algorithm is $O(\log k)$-competitive for the
online Steiner tree problem.




\section{Steiner Network}
\label{sec:steiner}
In this section, we consider the Steiner network problem and prove
Theorem \ref{thm:proper}. 
Recall that in the online Steiner network problem, a terminal pair
$(s_i,t_i)$ with requirement $R_i$ arrives at each online step and the
algorithm maintains a multigraph $H$ which contains $R_i$
edge-disjoint $(s_i,t_i)$-paths for each arrived terminal pair. Let
$\Rmax = \max_i R_i$ be the maximum requirement. Note that the Steiner
forest problem is a special case with $\Rmax = 1$.




\paragraph{Intuition.}

Consider a ``scaled'' online Steiner forest instance in which all
terminal pairs have the same requirement $R$. Since we can buy
multiple copies of an edge, we can just run the $O(\log
k)$-competitive Berman-Coulston algorithm and buy $R$ copies of every
edge it buys. The cost of the solution found is $O(\log k)\OPT$. This
suggests the following approach: decompose the online Steiner network
instance into $O(\log \Rmax)$ scaled Steiner forest subinstances and
use the Berman-Coulston algorithm on each of them.

Agrawal et al.~\cite{Agrawal:1995:TCA:204695.204699} applied this
approach in the offline setting (using their offline Steiner forest
algorithm) and obtained a $O(\log \Rmax)$-approximation, losing a
factor of $O(\log \Rmax)$ over the approximation factor of their
Steiner forest algorithm. There is a matching $\Omega(\log
\min\{\Rmax, k\})$ lower bound for this algorithm as
well. Intuitively, the $O(\log \Rmax)$ loss makes sense since each of
the $O(\log \Rmax)$ subinstances are treated
independently. Surprisingly, we show that in the online setting, this
approach is $O(\log k)$-competitive, losing only a constant factor
over the competitive ratio of the Berman-Coulston algorithm.



The analysis consists of two ingredients. The first is that for any
online Steiner forest instance, the cost of the Berman-Coulston
algorithm is at most a constant times the cost of the optimal solution
on any HST embedding of the terminals (Lemma \ref{lem:SF:cost}). The second is the following key
idea: for a Steiner network instance on a tree, the total cost of the
optimal solutions to the scaled Steiner forest subinstances is at most
a constant times the cost of the optimal solution to the original
Steiner network instance\footnote{In general graphs, we lose a factor
  of $O(\log
  \Rmax)$.}. 

\subsection{Steiner Forest}
\label{sec:forest}





We now show that given an online Steiner forest instance, the
Berman-Coulston algorithm (we call it Algorithm \ref{alg:SF}
henceforth) finds a solution whose cost is at most a constant times
the optimal solution on any HST embedding of the terminals. We remark
that this also gives a simpler analysis of the
algorithm. 

\paragraph{Algorithm.}
Algorithm \ref{alg:SF} proceeds as follows. When a terminal pair
$(s_i,t_i)$ arrives, it classifies $s_i$ and $t_i$ based on
their distance: $\class(s_i) = \class(t_i) = \lfloor \log d(s_i,t_i)
\rfloor$. The algorithm then proceeds in \emph{levels}, starting from
level $j = 0$ up to $\class(i)$; at level $j$, it connects each
terminal $v$ of class at least $j$ to $s_i$ if $d(s_i,v) < 2^{j+1}$,
or to $t_i$ if $d(t_i,v) < 2^{j+1}$. We emphasize that we always
consider distances according to the original metric $d$; we do not
contract added edges.
\begin{algorithm}
\caption{Berman-Coulston Algorithm for Steiner Forest}
  \begin{algorithmic}[1]
   \label{alg:SF}
   \STATE $H \leftarrow \emptyset$
   \WHILE {request $(s_i,t_i)$ arrives}
   \STATE Set class of $s_i$ and $t_i$ to be $\lfloor \log d(s_i,t_i) \rfloor$
   \FOR {level $j = 0$ \TO $\lfloor \log d(s_i,t_i) \rfloor$}
   \FOR {$v$ such that $\class(v) \geq j$ and $d(s_i,v) <
     2^{j+1}$} 
   \STATE Add $(s_i,v)$ to $H$
   \ENDFOR
   \FOR {$v$ such that $\class(v) \geq j$ and $d(t_i,v) <
     2^{j+1}$} 
   \STATE Add $(t_i,v)$ to $H$
   \ENDFOR
   \ENDFOR
   \ENDWHILE
 \end{algorithmic}
\end{algorithm}

\paragraph{Analysis.}
Let $X$ be the set of terminals that arrived and define $X_j$ to be
the set of terminals $t$ with $\class(t) \geq j$. Let $A_j$ be the set
of edges added in level $j$ of some iteration.
Note that $(u,v) \in A_j$ implies that $d(u,v) < 2^{j+1}$ and
both $u$ and $v$ have class at least $j$.

We need the following lemma for our charging scheme.

\begin{lemma}
  \label{lem:SF:sep}
  For each class $j$, let $\Scal_j$ be a collection of disjoint
  subsets of $X$ such that
  \begin{itemize}
  \item $\Scal_j$ covers $X_j$ and
  \item for each $S \in \Scal_j$, we have $d(u,v) < 2^j$ for $u,v \in S$.
  \end{itemize}
  Then we have $c(H) \leq \sum_j 2^{j+1}|\Scal_j|$.
\end{lemma}

\begin{proof}
  We have that $c(H) \leq \sum_j 2^{j+1}|A_j|$ and so it suffices to
  prove that $|A_j| \leq |\Scal_j|$ for each $j$. Fix $j$ and consider
  the following meta-graph: nodes correspond to $\Scal_j$; for each
  edge $(u,v) \in A_j$, there is a meta-edge between the nodes
  corresponding to sets containing $u$ and $v$. This is well-defined
  since $u,v \in X_j$ and $\Scal_j$ covers $X_j$. The meta-edges
  correspond to $A_j$ and the number of nodes is $|\Scal_j|$. We will
  show that the meta-graph is acyclic and so $|A_j| \leq |\Scal_j|$,
  as desired.

  Suppose, towards a contradiction, that there is a cycle in the
  meta-graph. Thus, there exists edges $(u_1, v_1), \ldots, (u_\ell,
  v_\ell) \in A_j$ and sets $S_1, \ldots, S_\ell \in \Scal_j$ such
  that $v_i, u_{i+1} \in S_{i+1}$ for each $i < \ell$ and $v_\ell, u_1
  \in S_1$. Suppose that $(u_\ell, v_\ell)$ was the last edge that was
  added by the algorithm. Since $(u_\ell, v_\ell) \in A_j$, it was
  added in level $j$ of some iteration. Consider the point in time
  right before the algorithm added $(u_\ell,v_\ell)$. For $1 \leq i
  \leq \ell$, the terminals $X_j \cap S_i$ that have arrived by this
  iteration are already connected by now since $\diam(S_i) < 2^j$ and
  they have class at least $j$. Therefore $u_\ell$ and $v_\ell$ are
  actually already connected at this time, and so the algorithm would
  not have added the edge $(u_\ell,v_\ell)$. This gives the desired
  contradiction and so the meta-graph is acyclic, as desired.
\end{proof}

\begin{lemma}
  \label{lem:SF:cost}
  Define the function $f :2^X \rightarrow \{0,1\}$ as follows: $f(S) =
  1$ if $S$ separates a terminal pair. For any HST
  embedding $T$ of $(X,d)$, we have $c(H) \leq 4\sum_j \sum_{C \in
    \C_j} 2^{j-1}f(C) = 4\OPT(T)$.
\end{lemma}

\begin{proof}
  Let $T$ be a HST embedding. Define $\Scal_j = \{C \in \C_j(T) : C
  \cap X_j \neq \emptyset\}$, i.e. $\Scal_j$ is the collection of
  level-$j$ cuts that contain a terminal of class at least $j$. Since
  $\Scal_j$ satisfies the conditions of Lemma \ref{lem:SF:sep}, we
  have that \[c(H) \leq \sum_j 2^{j+1}|\Scal_j|.\]  

  Next, we lower bound $\OPT(T)$ in terms of its hierarchical
  decomposition $\C(T)$. 
  For any level-$j$ cut $C_e \in \C_j(T)$,
  if $(s_i,t_i) \in \delta(C_e)$ then $e$ lies on the $(s_i,t_i)$ path
  in $T$ so the optimal solution on $T$ buys $e$. Thus, we have that
  \[\OPT(T) = \sum_j \sum_{C \in \C_j(T)} 2^{j-1}  f(C).\] 


  Since $c(H) \leq \sum_j 2^{j+1}|\Scal_j|$ and $\C_j(T) \supseteq
  \Scal_j$, it suffices to prove that $f(C) = 1$ for $C \in
  \Scal_j$. Fix a cut $C \in \Scal_j$. By definition of $\Scal_j$,
  there exists a terminal $s_i \in C$ with $\class(s_i) \geq j$. Since
  $\class(s_i) \geq j$, we have that $d(s_i,t_i) \geq 2^j$ and so $t_i
  \notin C$ as the diameter of a level-$j$ cut is less than
  $2^j$. Thus, $C$ separates the terminal pair $(s_i,t_i)$ and so
  $f(C) = 1$, as desired. Now we have $\sum_{C \in \C_j(T)} f(C) \geq
  |\Scal_j|$ for each level $j$ and this completes the proof of the
  lemma.
\end{proof}

Observe that this lemma together with Corollary \ref{cor:main} implies
that Algorithm \ref{alg:SF} is $O(\log k)$-competitive for the online
Steiner forest problem. Lemma \ref{lem:SF:cost} says that the cost of
Algorithm \ref{alg:SF} is at most $O(1)\OPT(T)$ for any HST embedding
$T$ of $(X,d)$. Furthermore, Corollary \ref{cor:main} implies that
there exists a HST embedding $T^*$ such that $\OPT(T^*) \leq O(\log
k)\OPT$. Thus, Algorithm \ref{alg:SF} is $O(\log k)$-competitive for
the online Steiner forest problem.

\subsection{From Steiner Forest to Steiner Network}
\label{sec:SN}

We now state our Steiner network algorithm formally and show that it
is $O(\log k)$-competitive.
\paragraph{Algorithm.}
For ease of exposition, we first assume that we are given
the maximum requirement $\Rmax$ and remove this assumption later
on. 
We run $\lfloor \log \Rmax \rfloor$ instantiations of Algorithm
\ref{alg:SF} (the Berman-Coulston Steiner forest algorithm); when we
receive a terminal pair with requirement $R_i \in [2^\ell, 2^{\ell +
  1})$, we pass the pair to the $\ell$-th instantiation and buy
$2^{\ell+1}$ copies of each edge bought by that instantiation. The
assumption that we are given $\Rmax$ can be removed by starting a new
instantiation of Algorithm \ref{alg:SF} when we receive a terminal
pair whose requirement is higher than all previous requirements.

\paragraph{Analysis.} We maintain a feasible solution so it remains to
bound its cost.  Let $X$ be the set of terminals that arrived and
$H_\ell$ be the final subgraph of the $\ell$-th instantiation of
Algorithm \ref{alg:SF}. The cost of our solution is $\sum_\ell
2^{\ell+1}c(H_\ell)$. We now show that this is at most a constant
times the cost of the optimal solution on any HST embedding of the
terminals.

\begin{lemma}
  \label{lem:SN:cost}
  $\sum_\ell 2^{\ell+1}c(H_\ell) \leq O(1)\OPT(T)$ for any
  HST embedding $T$ of $(X,d)$.
\end{lemma}

\begin{proof}
  Fix a HST embedding $T$ of $(X,d)$. Define the function $f : 2^X
  \rightarrow \mathbb{N}$ as follows: for $S \subseteq X$, we have
  $f(S) = \max_{i : (s_i,t_i) \in \delta(S)} R_i$.  For any level-$j$
  cut $C_e \in \C_j(T)$, if $(s_i,t_i) \in \delta(C_e)$ then $e$ lies
  on the $(s_i,t_i)$ path in $T$ so the optimal solution on $T$ buys
  $f(C_e)$ copies of $e$. Summing over all cuts of the hierarchical
  decomposition gives us 
  \[\OPT(T) = \sum_j \sum_{C \in \C_j(T)} 2^{j-1} f(C).\]
  
  For each $0 \leq \ell \leq \lfloor \log \Rmax \rfloor$, define the
  subset of terminal pairs $D_\ell = \{(s_i, t_i) : R_i \in [2^\ell,
  2^{\ell + 1})\}$ and the function $f_\ell :2^X \rightarrow \{0,1\}$
  as follows: $f_\ell(S) = 1$ if $S \subseteq X$ separates a terminal
  pair of $D_\ell$. 
  Since $H_\ell$ is the output of
  Algorithm \ref{alg:SF} when run on the subsequence of terminal pairs $D_\ell$, Lemma
  \ref{lem:SF:cost} implies that
  \begin{align*}\sum_\ell 2^{\ell+1}c(H_\ell) 
    &\leq 4 \sum_\ell 2^{\ell+1}
    \left(\sum_j \sum_{C \in \C_j(T)} 2^{j-1} f_\ell(C)\right) \\
    &= 4\sum_j
    \sum_{C \in \C_j(T)} 2^{j-1}\left(\sum_\ell
      2^{\ell+1}f_\ell(C)\right).
  \end{align*}

  Fix a cut $C$. Since each terminal pair $(s_i,t_i) \in D_\ell$ has
  requirement $R_i \in [2^\ell, 2^{\ell+1})$ and $f_\ell(C) = 1$ if
  $C$ separates a terminal pair of $D_\ell$, we have $\sum_\ell
  2^{\ell + 1} f_\ell(C) \leq 4 \max_{i : (s_i,t_i) \in \delta(C)} R_i
  = 4f(C)$. Thus, we have
  \begin{align*}
    4\sum_j \sum_{C \in \C_j(T)}
    2^{j-1}\left(\sum_\ell 2^{\ell+1}f_\ell(C)\right) 
    \leq 16\sum_j \sum_{C \in \C_j(T)} 2^{j-1} f(C) 
    = 16 \OPT(T).
  \end{align*}
\end{proof}

By Corollary \ref{cor:main}, there exists a HST embedding $T^*$ such
that $\OPT(T^*) \leq O(\log k)\OPT$. So Lemma \ref{lem:SN:cost}
implies that the cost of our algorithm is at most $O(\log
k)\OPT$. Thus, our algorithm is $O(\log k)$-competitive for Steiner
network and this proves Theorem \ref{thm:proper}.\footnote{We remark
  that one can apply the usual dual-fitting analysis, but the upper
  bound it gives is proportional to the number of relevant distance
  scales. This is $O(\log (k\Rmax))$ for this problem and can
  deteriorate to $O(k)$ since $\Rmax$ can be as large as $2^k$.}

\section{Rent-or-Buy}
\label{sec:ROB}

Next, we consider the rent-or-buy problem and prove Theorem
\ref{thm:MROB}. We illustrate the key ideas by proving the
single-source version of the theorem in Section \ref{sec:SROB}. Then
we extend these ideas to the multi-source setting in Section
\ref{sec:MROB}. 

\subsection{Single-Source Rent-or-Buy}
\label{sec:SROB}
We prove the following theorem.

\begin{theorem}
  \label{thm:SROB}
  There exists a $O(\log k)$-competitive algorithm for the online
  single-source rent-or-buy problem.
\end{theorem}
 
We recall the problem statement. The algorithm is initially given a
root terminal $r$ and a parameter $M \geq 0$. The algorithm maintains a
subgraph $H$ of bought edges. When a terminal $i$ arrives, the
algorithm buys zero or more edges and rents edges $Q_i$ such that $H
\cup Q_i$ connects $i$ and $r$. The total cost of the algorithm is
$Mc(H) + \sum_i c(Q_i)$.  We call $Mc(H)$ the \emph{buy cost} and
$c(Q_i)$ the \emph{rent cost} of terminal $i$.

\paragraph{Intuition.}
Our algorithm maintains that the subgraph $H$ of bought edges is a
connected subgraph containing the root; essentially it either rents or
buys the shortest edge from the current terminal to $H$. Roughly
speaking, our algorithm buys if there are $M$ terminals nearby with
sufficiently large rent costs, and rents otherwise. This allows us to
charge the buy cost to the total rent cost and we then charge the
total rent cost to edges of any HST embedding $T$ by showing that
there are few terminals with large rent cost in a small neighborhood.

\paragraph{Algorithm.}
Algorithm \ref{alg:SROB} maintains a set of \emph{buy terminals} $Z$
($Z$ includes $r$) which are connected by $H$ to the root; all other
terminals are called \emph{rent terminals} and denoted by $R$. When a
terminal $i$ arrives, let $z \in Z$ be the closest buy terminal and
define $a_i = d(i,z)$. 
As in the previous section, we define $\class(i) = j$ if $a_i \in
[2^j, 2^{j+1})$. The algorithm considers the set of class-$j$ rent
terminals $R_j$ of distance less than $2^{j-1}$ to $i$. Call this
the \emph{witness set} of $i$ and denote it by $W(i)$. The algorithm 
buys the edge $(i,z)$ if $|W(i)| \geq M$ and rents it
otherwise. Terminal $i$ becomes a buy terminal in the former case and
a rent terminal in the latter case.


\begin{algorithm}
   \caption{Algorithm for Online Single-Source Rent-or-Buy}
  \begin{algorithmic}[1]
    \label{alg:SROB}
    \STATE $Z \leftarrow \{r\}; R_j \leftarrow
    \emptyset; H \leftarrow \emptyset$
    \WHILE {terminal $i$ arrives}
    \STATE Let $z$ be closest terminal in $Z$ to $i$ 
     and set $j = \lfloor \log d(i,z) \rfloor$
    \STATE $W(i) \leftarrow \{i' \in R_j : d(i,i') < 2^{j-1}\}$
    \IF {$|W(i)| \geq M$}
    \STATE Add $i$ to $Z$ and buy $(i,z)$, i.e. $H \leftarrow H \cup \{(i,z)\}$
    \ELSE
    \STATE Add $i$ to $R_j$ and rent $(i,z)$, i.e. $Q_i \leftarrow \{(i,z)\}$ 
    \ENDIF
    \ENDWHILE
  \end{algorithmic}
\end{algorithm}
 
\paragraph{Analysis.}
Let $X$ be the set of all terminals.  For each buy terminal $z \in Z$,
the algorithm incurs a buy cost of $Ma_z$; for each rent terminal $i
\in R$, the algorithm incurs a rent cost of $a_i$. Thus, the total
cost of the algorithm is $\sum_{z \in Z} Ma_z + \sum_{i \in R}
a_i$. Our goal in the analysis is to show that this is at most a
constant times $\OPT(T)$ for any HST embedding $T$ of $(X,d)$.


First, we define cost shares for each terminal. For each rent terminal
$i \in R_j$, we define $i$'s cost share to be $2^{j+1}$, i.e. $a_i$
rounded up to the next power of $2$. The total
cost share is $\sum_j 2^{j+1}|R_j|$. We now show that the cost of the
algorithm is at most twice the total cost share.

\begin{lemma}
  \label{lem:SROB:share}
  $\sum_{z \in Z} Ma_z + \sum_{i \in R} a_i \leq 2 \sum_j 2^{j+1}|R_j|$.
\end{lemma}

\begin{proof}
  Let $Z_j \subseteq Z$ be the set of class-$j$ buy terminals. We have
  $\sum_{i \in R} a_i \leq \sum_j 2^{j+1}|R_j|$ and $\sum_{z \in Z}
  Ma_z \leq \sum_j M2^{j+1}|Z_j|$. We now show that $M|Z_j| \leq
  |R_j|$ for each class $j$.
  
  Fix a class $j$. The witness set $W(z)$ of $z \in Z_j$ satisfies the
  following properties: (1) $|W(z)| \geq M$; (2) $W(z) \subseteq R_j$; (3)
  $d(i,z) < 2^{j-1}$ for $i \in W(z)$. The first implies that
  $M|Z_j| \leq \sum_{z \in Z_j} |W(z)|$. We claim that $d(z,z') \geq
  2^j$ for $z,z'\in Z_j$. This completes the proof since together with
  the second and third properties, we have that the witness sets of
  $Z_j$ are disjoint subsets of $R_j$ and so $\sum_{z \in Z_j} |W(z)|
  \leq |R_j|$.

  Now we prove the claim. Observe that $H$ is the subgraph produced by
  the greedy online Steiner tree algorithm if it were run on the
  subsequence of buy terminals $Z$ and for each $z \in Z$, we have
  that $a_z$ is exactly the distance from $z$ to the nearest
  previously-arrived buy terminal. Thus, we can apply Lemma
  \ref{lem:r-sep} and get that $d(z,z') \geq 2^j$ for any $z,z' \in
  Z_j$, proving the claim. Putting all of the above together, we get
  $\sum_{z \in Z} Ma_z \leq \sum_j 2^{j+1}|R_j|$. This finishes the
  proof of the lemma.
\end{proof}

Next, we show that the total cost share is at most a constant times
the cost of the optimal solution on any HST embedding of the
terminals. 






\begin{lemma}
  \label{lem:SROB:cost}
  $\sum_j 2^{j+1}|R_j| \leq O(1)\OPT(T)$ for any HST embedding $T$ of
  $(X,d)$. 
\end{lemma}

\begin{proof}
  Let $T$ be a HST embedding of $(X,d)$. 
  We begin by lower bounding $\OPT(T)$ in terms of $T$'s cuts.  Fix a
  level-$j$ cut $C_e \in \C_j(T)$. If $r \notin C_e$, then $e$ lies on
  the path between $i$ and $r$ for each terminal $i \in C_e$. Since $e
  \in E_j(T)$ and has length $2^{j-1}$, the optimal solution on $T$
  either rents $e$ for each terminal in $C$ at a cost of $2^{j-1}|C|$
  or buys it at a cost of $2^{j-1}M$. Summing over all cuts of the
  hierarchical decomposition gives us the following lower bound
\[\OPT(T) \geq \sum_j \sum_{C \in \C_j(T) : r \notin C}
  2^{j-1}\min\{M,|C|\}.\]

  For each rent terminal $i \in R_j$, we charge
  $2^{j+1}$ to the unique level-$(j-1)$ cut\footnote{
    We can extend $T$ with an additional level of terminal singletons
    to accomodate charging against level $j = -1$. This only increases
    $\OPT(T)$ by at most a constant.
  } containing $i$. Each
  level-$j$ cut $C \in \C_j(T)$ is charged $2^{j+2} |R_{j+1} \cap C|$
  since it receives a charge of $2^{j+2}$ for each class-$(j+1)$ rent
  terminal in it. So overall, we have \[\sum_j 2^{j+1}|R_j| = \sum_j
  \sum_{C \in \C_j(T)} 2^{j+2} |R_{j+1} \cap C|.\]

  It remains to prove the following claim: for each level-$j$ cut $C
  \in \C_j(T)$, we have $|R_{j+1} \cap C| = 0$ if $r \in C$ and
  $|R_{j+1} \cap C| \leq \min\{M,|C|\}$ if $r \notin C$. Suppose $r
  \in C$. Since $a_i < d(i,r) < 2^{j}$ for all $i \in C$, there cannot
  be any class-($j+1$) terminal in $C$ and so $R_{j+1} \cap C =
  \emptyset$. Now consider the case $r \notin C$. Since $|R_{j+1} \cap
  C| \leq |C|$, it suffices to prove that $|R_{j+1} \cap C| \leq
  M$. Suppose, towards a contradiction, that $|R_{j+1} \cap C| > M$
  and let $i$ be the last-arriving terminal of $R_{j+1} \cap C$.  The
  terminals of $R_{j+1} \cap C \setminus \{i\}$ arrive before $i$, are
  each of distance less than $2^j$ from $i$ (diameter of $C$ is less
  than $2^j$) and of the same class as $i$, so they are part of $i$'s
  witness set $W(i)$. Since $|R_{j+1} \cap C| > M$, we have that
  $|W(i)| \geq |R_{j+1} \cap C \setminus \{i\}| \geq M$. Thus, $i$
  would have been a buy terminal but this contradicts the assumption
  that $i \in R_{j+1}$. Therefore, we have $|R_{j+1} \cap C| \leq M$
  as desired. This completes the proof of the claim and so we get
  $\sum_j 2^{j+1}|R_j| \leq 8\OPT(T)$.
\end{proof}

Now we put all of the above together. Lemmas \ref{lem:SROB:share} and
\ref{lem:SROB:cost} imply that the cost of Algorithm \ref{alg:SROB} is
at most $O(1)\OPT(T)$ for any HST embedding $T$ of
$(X,d)$. Furthermore, by Corollary
\ref{cor:main}, there exists a HST embedding $T^*$ such
that $\OPT(T^*) \leq O(\log k)\OPT$. 
Thus the algorithm is $O(\log k)$-competitive for single-source
rent-or-buy and this proves Theorem \ref{thm:SROB}.

\subsection{Multi-Source Rent-or-Buy}
\label{sec:MROB}


We move on to the multi-source setting. 
Recall the problem statement. The algorithm is given a parameter $M
\geq 0$ initially, and it maintains a subgraph $H$ of bought edges
online. When a terminal pair $(s_i, t_i)$ arrives, the algorithm buys
zero or more edges and rents edges $Q_i$ such that $H \cup Q_i$
connects $s_i$ and $t_i$.  The total cost of the algorithm is
$Mc(H) + \sum_i c(Q_i)$. We call $Mc(H)$ the \emph{buy cost} and
$c(Q_i)$ the \emph{rent cost} of terminal pair $(s_i,t_i)$. 

\paragraph{Intuition.}
The high-level idea is similar to the single-source rent-or-buy
algorithm in Section \ref{sec:SROB}.  For each terminal pair
$(s_i,t_i)$, our algorithm will either rent the edge $(s_i,t_i)$ or
buy edges such that $s_i$ and $t_i$ are connected in $H$. Roughly
speaking, the algorithm buys edges if there are at least $M$ endpoints
of terminal pairs with sufficiently large rent costs near each of
$s_i$ and $t_i$, and rents the edge $(s_i,t_i)$ otherwise. The main
difference with the single-source case is that $H$ is a Steiner forest
connecting a subset of the terminal pairs. Thus, we use the
Berman-Coulston algorithm for online Steiner forest
\cite{Berman:1997:OAS:258533.258618} (described in Section
\ref{sec:forest}) to determine $H$.



\paragraph{Algorithm.}
For each terminal pair $(s_i,t_i)$, we define $a_i = d(s_i,t_i)$ and
classify $s_i$ and $t_i$ based on $a_i$: $\class(s_i) = \class(t_i) =
j$ if $d(s_i,t_i) \in [2^j, 2^{j+1})$. As terminal pairs arrive
online, the algorithm designates some of the terminals as \emph{rent
  terminals}. When $(s_i,t_i)$ arrives, the algorithm considers the
sets $W(s_i)$ and $W(t_i)$ of class-$j$ rent terminals that are of
distance less than $2^{j-2}$ to $s_i$ and $t_i$, respectively. We call
these the \emph{witness sets} of $s_i$ and $t_i$. If $|W(s_i)| \geq M$
and $|W(t_i)| \geq M$, the algorithm buys edges such that $s_i$ and
$t_i$ are connected in $H$; otherwise, it rents the edge $(s_i,t_i)$
and designates exactly one of $s_i$ or $t_i$ to be a rent terminal. We
say that $(s_i,t_i)$ is a \emph{buy pair} if the edge is bought and a
\emph{rent pair} otherwise.

In the following, denote by $R_j$ the set of class-$j$ rent terminals
that have arrived so far.
\begin{algorithm}
  \caption{Algorithm for Online Multi-Source Rent-or-Buy}
  \begin{algorithmic}[1]
    \label{alg:MROB}
    \STATE $H \leftarrow \emptyset; R_j \leftarrow \emptyset$
    \WHILE {request $(s_i, t_i)$ arrives}
    \STATE Set $j = \class(s_i) = \class(t_i) = \lfloor \log
    d(s_i,t_i) \rfloor$ 
    \STATE $W(s_i) \leftarrow \{v \in R_j : d(s_i,v) < 2^{j-2}\}$
    \STATE $W(t_i) \leftarrow \{v \in R_j : d(t_i,v) < 2^{j-2}\}$
    \IF { $|W(s_i)| < M$}
    \STATE Rent $(s_i,t_i)$, i.e. $Q_i \leftarrow \{(s_i,t_i)\}$ 
    \STATE Add $s_i$ to $R_j$ 
    \ELSIF {$|W(t_i)| < M$}
    \STATE Rent $(s_i,t_i)$, i.e. $Q_i \leftarrow \{(s_i,t_i)\}$  
    \STATE Add $t_i$ to $R_j$. 
    \ELSE
    \STATE Pass $(s_i,t_i)$ to Algorithm \ref{alg:SF} and buy each
    edge that it buys; add bought edges to $H$
    \ENDIF
    \ENDWHILE
  \end{algorithmic}
\end{algorithm}

\paragraph{Analysis.}
Let $X$ be the set of terminals that arrived, $R$ be the set of all
rent terminals and $Z$ be the set of terminals that were passed to
Algorithm \ref{alg:SF} (i.e. $Z$ is the set of endpoints of buy
pairs). We call $Z$ the set of \emph{buy terminals}. Note that $R$
contains exactly one endpoint of each rent pair; we rename the
terminals so that if $(s_i,t_i)$ is a rent pair, then $t_i \in R$. The
cost of the algorithm is $Mc(H) + \sum_{t_i \in R} d(s_i,t_i)$. The
analysis proceeds by charging the cost of the algorithm against the
cost of the optimal solution on any HST embedding of the terminals.

We first define cost shares for each terminal. For each $t_i \in R_j$,
we define $t_i$'s cost share to be $2^{j+1}$. The total cost share is
$\sum_j 2^{j+1}|R_j|$. We now show that the cost of the algorithm is
at most twice the total cost share.


\begin{lemma}
  \label{lem:MROB:share}
  $Mc(H) + \sum_{t_i \in R} d(s_i,t_i) \leq 2 \sum_j 2^{j+1}|R_j|$.
\end{lemma}

\begin{proof}
  Since $t_i$'s cost share is at least $d(s_i,t_i)$, we have that
  $\sum_{t_i \in R} d(s_i,t_i) \leq \sum_j 2^{j+1}|R_j|$. Next, we
  bound $Mc(H)$. Let $Z_j \subseteq Z$ be the set of buy terminals
  with class at least $j$. Define $Z'_j \subseteq Z_j$ to be the
  maximal subset of $Z_j$ such that $d(u,v) \geq 2^{j-1}$ for all $u,v
  \in Z'_j$. We will show that $Mc(H) \leq \sum_j 2^{j+1}|R_j|$ in two
  steps: first we prove that $Mc(H) \leq \sum_j 2^{j+1}M|Z'_j|$ and then
  prove that $M|Z'_j| \leq |R_j|$ for each $j$.


  We partition $Z_j$ as follows: assign each $v \in Z_j$ to the
  closest terminal in $Z'_j$, breaking ties arbitrarily, and define
  $S_u$ to be the set of terminals assigned to $u$. Observe that the
  diameter of $S_u$ is less than $2^j$ for all $u \in Z'_j$. Define
  $\Scal_j = \{S_u\}_{u \in Z'_j}$. Since $\Scal_j$ satisfies the
  conditions of Lemma \ref{lem:SF:sep}, we get that $Mc(H) \leq
  \sum_j2^{j+1}M|\Scal_j| = \sum_j2^{j+1}M|Z'_j|$.
  
  Next we show that $M|Z'_j| \leq 2^{j+1}|R_j|$ for each $j$.  Fix
  $j$. The witness set $W(z)$ of $z \in Z'_j$ satisfies the following
  properties: (1) $|W(z)| \geq M$; (2) $W(z) \subseteq R_j$; (3) $d(i,z) <
  2^{j-2}$ for $i \in W(z)$. The first implies that $M|Z_j| \leq
  \sum_{z \in Z_j} |W(z)|$. We have $d(z,z') \geq 2^{j-1}$ for $z, z'
  \in Z'_j$ by definition, so with the second and third properties, we
  get that the witness sets of $Z'_j$ are disjoint subsets of $R_j$,
  and so $\sum_{z \in Z'_j} |W(z)| \leq |R_j|$. Putting all of the
  above together, we have $Mc(H) \leq \sum_j 2^{j+1} |R_j|$.
\end{proof}



Next, we show that we can charge the cost shares against the optimal
solution on any HST embedding $T$. 


\begin{lemma}
  \label{lem:MROB:C}
  $\sum_j 2^{j+1}|R_j| \leq O(1)\OPT(T)$ for all HST
  embeddings $T$ of $(X,d)$.
\end{lemma}

\begin{proof}
  Let $T$ be a HST embedding of $(X,d)$. We begin by expressing
  $\OPT(T)$ in terms of $T$'s cuts. Define $D(S) = \{(s_i,t_i) :
  (s_i,t_i)\in\delta(S)\}$ for each vertex subset $S \subseteq X$.
  Consider a level-$j$ cut $C_e \in \C_j(T)$. By definition, $e \in
  E_j(T)$ and has length $2^{j-1}$. If $(s_i,t_i) \in \delta(C_e)$,
  then $e$ lies on the $(s_i,t_i)$ path in $T$. Thus, the optimal
  solution on $T$ either rents $e$ for each terminal pair in $D(C_e)$
  at a cost of $2^{j-1}|D(C_e)|$ or buys it at a cost of $2^{j-1}M$
  and so \[\OPT(T) = \sum_j \sum_{C \in C_j(T)} 2^{j-1} \cdot
  \min\{M, |D(C)|\}.\]

  For each rent terminal $t_i \in R_j$, we charge its cost share
  $2^{j+1}$ to the level-($j-2$) cut $C \in \C_{j-2}(T)$ containing
  $i$. Each level-$j$ cut $C \in \C_j(T)$ receives a charge of
  $2^{j+3} |R_{j+2} \cap C|$ and so
  \[\sum_j 2^{j+1} |R_j| = \sum_j 2^{j+3} \sum_{C \in \C_j(T)} |R_{j+2}
  \cap C|.\] 

  It remains to prove the following claim: for each level-$j$ cut $C
  \in \C_j(T)$, we have $|R_{j+2} \cap C| \leq \min\{M, |D(C)|\}$. 
  Since $C$ has diameter less than $2^j$ and $d(s_i,t_i) \geq 2^{j+2}$
  for each $t_i \in R_{j+2}$, we have $|R_{j+2} \cap C| \leq
  |D(C)|$. So now we prove that $|R_{j+2} \cap C| \leq M$. Suppose,
  towards a contradiction, that $|R_{j+2} \cap C| > M$ and let $t_i$
  be the last-arriving terminal of $R_{j+2} \cap C$.  The terminals of
  $R_{j+2} \cap C \setminus \{t_i\}$ arrive before $t_i$, are each of
  distance less than $2^j$ from $t_i$ (diameter of $C$ is less than
  $2^j$) and of the same class as $t_i$, so they are part of $t_i$'s
  witness set $W(t_i)$. Since $|R_{j+2} \cap C| > M$, we have that
  $|W(t_i)| \geq |R_{j+2} \cap C \setminus \{t_i\}| \geq M$. Thus,
  $t_i$ would have been a buy terminal but this contradicts the
  assumption that $t_i \in R_{j+2}$. Therefore, we have $|R_{j+2} \cap
  C| \leq M$ as desired. This completes the proof of the claim and
  so 
  $\sum_j 2^{j+1}|R_j| \leq 16\OPT(T)$.  \end{proof}

Now, Lemmas \ref{lem:MROB:share} and \ref{lem:MROB:C} imply that the
cost of Algorithm \ref{alg:MROB} is at most $O(1)\OPT(T)$ for any HST
embedding $T$ of $(X,d)$. Furthermore, by Corollary \ref{cor:main},
there exists a HST embedding $T^*$ such that $\OPT(T^*)
\leq O(\log k)\OPT$. Thus the algorithm is $O(\log k)$-competitive for
the online rent-or-buy problem and this proves Theorem \ref{thm:MROB}.









\section{Connected Facility Location}
\label{sec:CFL}
In this section, we consider the connected facility location problem
and prove Theorem \ref{thm:CFL}. We recall the problem statement. At
the beginning, the algorithm is given a parameter $M \geq 0$, a set of
\emph{facilities} $F \subseteq V$ and facility opening costs $f_x$ for
each facility $x \in F$. There is a designated root facility $r \in F$
with zero opening cost. The algorithm maintains a set of open
facilities $F'$ and a subgraph $H$ connecting $F'$ and $r$. When a
client $i$ arrives, the algorithm may open a new facility, and then
assigns $i$ to some open facility $\sigma(i) \in F'$. The cost of the
algorithm is $\sum_{x \in F'} f_x + Mc(H) + \sum_i d(i,
\sigma(i))$. We call $Mc(H)$ the \emph{Steiner cost} and
$d(i,\sigma(i))$ the \emph{assignment cost} of client $i$. The special
case when $M = 0$ (i.e. open facilities need not be connected) is
called the online facility location problem.

\paragraph{Intuition.} 
The connected facility location problem shares a similar cost structure
with single-source rent-or-buy. Here, the ``buy cost'' consists of the
facility opening cost and the Steiner cost, the ``rent cost'' is the
assignment cost. At a high-level, we use a similar strategy --- we
only open a new facility if we can pay for the facility opening cost
and the additional Steiner cost using the assignment costs of nearby
clients. In order to do this, we use an online facility location
algorithm together with an adaptation of our single-source rent-or-buy
algorithm. The former tells us if we can pay for the opening cost and
the latter tells us if we can pay for the additional Steiner cost
incurred in connecting the newly open facility to $r$.

In the analysis, we will not be able to show that the cost of the
algorithm is at most $O(1)\OPT(T)$ for any HST embedding $T$ of the
clients because there is a lower bound of $\Omega(\frac{\log
  k}{\log \log k})$ for the facility location problem even on HSTs
\cite{fotakis2008competitive}. In a sense, routing is easy on trees
but choosing a good set of facilities to open online remains
difficult.  Instead, we will use the fact that the connected facility
location instance induces a rent-or-buy instance and a facility
location instance on the same metric, and that there is a $O(\log
k)$-competitive facility location
algorithm~\cite{fotakis2007primal}\footnote{We can also apply our
  techniques to show that the algorithm of \cite{fotakis2007primal} is
  $O(\log k)$-competitive but we omit the discussion as the analysis
  differs greatly from the analyses in this paper.}. Our analysis
charges part of the cost to the cost of the facility location
algorithm and the rest to the optimal rent-or-buy solution.

\paragraph{Algorithm.}
Since the root facility has zero opening cost, our algorithm starts by
opening the root facility.  We run Fotakis's $O(\log k)$-competitive
algorithm~\cite{fotakis2007primal} in parallel; call it $\OFL$. We
denote its open facilities by $\vF$ and its assignments by $\vassign$,
and call them \emph{virtual facilities} and \emph{virtual
  assignments}, respectively. Algorithm \ref{alg:CFL} will only open a
facility if it was already opened by $\OFL$. i.e. its open facilities
$F'$ is a subset of the virtual facilities $\vF$. (We assume
w.l.o.g. that $\OFL$ opens the root facility as well.)

When a client $i$ arrives, let $x \in F'$ be the nearest open
facility. Define $a_i = d(i,x)$ and $\class(i) = j$ if $a_i \in [2^j,
2^{j+1})$. We first pass the client to $\OFL$, which may open a new
virtual facility and then it assigns $i$ to the nearest virtual
facility $\vassign(i)$. We will either assign $i$ to $\vassign(i)$ or
the nearest open facility $x$. (Note that $F' \subseteq \vF$ so $i$ is
closer to $\vassign(i)$ than $x$.) If $x$ is not much further than
$\vassign(i)$ --- in particular, if $d(i,x) \leq 4d(i,\vassign(i))$
--- then we assign $i$ to $x$ since we can later charge the assignment
cost $d(i,x)$ to $\OFL$'s assignment cost $d(i,\vassign(i))$. In this
case, we call $i$ a \emph{virtual client}. Otherwise, we consider
opening $\vassign(i)$ and assigning $i$ to it. We define the
\emph{witness set} $W(i)$ to be the class-$j$ clients that are of
distance at most $2^{j-2}$ from $i$. If there are at least $M$
witnesses, then we open $\vassign(i)$, assign $i$ to $\vassign(i)$ and
connect $\vassign(i)$ to the root via $x$ (i.e. we add the edge
$(\vassign(i), x)$ to $H$); we call $i$ a \emph{buy client} in this
case. Otherwise, then we simply assign $i$ to $x$
and call $i$ a \emph{rent client}.

In the description of Algorithm \ref{alg:CFL} below, we use
$Q_j,Z_j,R_j$ to keep track of the class-$j$ virtual clients, buy
clients and rent clients, resp.
\begin{algorithm}
  \caption{Algorithm for Online Connected Facility Location}
   \begin{algorithmic}[1]
     \label{alg:CFL}
    \STATE Initialize $F' \leftarrow \{r\}; H \leftarrow \emptyset;
    Q_j \leftarrow \emptyset; Z_j \leftarrow
    \emptyset; R_j \leftarrow \emptyset;$
    \WHILE {client $i$ arrives}
    \STATE Pass $i$ to $\OFL$ and update virtual solution $\vF,\vassign$
    \STATE Let $x \in F'$ be nearest open facility to $i$ and 
    set $j = \lfloor \log d(i,x)
    \rfloor$ 
    \STATE $W(i) \leftarrow \{i' \in R_j : d(i',i) <
    2^{j-2}\}$
    \IF {$a_i \leq 4d(i, \vassign(i))$}
    \STATE Assign $i$ to $x$ and add $i$ to $Q_j$
    \ELSIF {$|W(i)| \geq M$}  
    \STATE Open $\vassign(i)$ and add $\vassign(i)$ to $F'$
    \STATE Add edge $(\vassign(i),x)$ to $H$
    \STATE Assign $i$ to $\vassign(i)$ and add $i$ to $Z_j$
    \ELSE
    \STATE Assign $i$ to $x$ and add $i$ to $R_j$
    \ENDIF
    \ENDWHILE
  \end{algorithmic}
\end{algorithm}

\paragraph{Analysis.}
Let $X$ be the set of clients and the root, and $k$ be the number of clients. 
Consider a facility location instance over metric $(V,d)$ with the
same facilities $F$ and clients $X$ as well as a rent-or-buy instance
over metric $(V,d)$ with terminals $X$ and root $r$.  Let $\OPTFL$ and
$\OPTROB$ be the cost of the respective optimal solutions. A
feasible solution to the connected facility location instance is
feasible for both the facility location and rent-or-buy instances so we
have the following lemma.

\begin{lemma}
  \label{lem:CFL:sub}
  $\OPTFL \leq \OPT$ and $\OPTROB \leq \OPT$.
\end{lemma}

Let $Q,Z,R$ be the set of virtual, buy, and rent clients,
respectively. Note that a virtual or rent client $i \in Q \cup R$ has
assignment cost $d(i,\sigma(i)) = a_i$ and a buy client $i \in Z$ has
assignment cost $d(i,\sigma(i)) = d(i, \vassign(i))$. Thus the cost
of the algorithm is \[\sum_{x \in F'} f_x + Mc(H) + \sum_{i \in Q \cup
  R} a_i + \sum_{i \in Z} d(i,\vassign(i)).\] We first charge the
opening cost as well as the assignment cost of virtual and buy clients
to the cost of the virtual solution. Then we charge the Steiner cost
and the assignment cost of rent clients to $\OPTROB$.

\begin{lemma}
  \label{lem:CFL:FL}
  $\sum_{x \in F'} f_x + \sum_{i \in Q} a_i + \sum_{i \in Z} d(i,
  \vassign(i)) \leq O(\log k)\OPT$.
\end{lemma}

\begin{proof}
  Our algorithm only opens a facility if it was already opened by
  $\OFL$ so $\sum_{x \in F'} f_x \leq \sum_{x \in \vF}
  f_x$. Furthermore, the assignment cost of a virtual client $i \in Q$
  is $a_i \leq 4d(i,\vassign(i))$, and for each buy client $i \in Z$
  we have $d(i,\sigma(i)) = d(i,\vassign(i))$. Therefore, we have
  \[\sum_{x \in F'} f_x + \sum_{i \in Q}
  a_i + \sum_{i \in Z} d(i, \vassign(i)) \leq \sum_{x \in \vF} f_x +
  4\sum_i d(i, \vassign(i)).\]

  Since $\OFL$ is a $O(\log k)$-competitive algorithm for facility
  location, the cost of its virtual solution is $\sum_{x \in \vF} f_x
  + \sum_i d(i, \vassign(i)) \leq O(\log k)\OPTFL$. The lemma now
  follows from Lemma \ref{lem:CFL:sub}. 
\end{proof}



Next, we show that the Steiner cost and the assignment cost of rent
clients $Mc(H) + \sum_{i \in R} a_i$ is at most a constant times the
cost of the optimal rent-or-buy solution on any HST embedding of the
terminals. 
%
%
This part of the analysis is analogous to that in Section
\ref{sec:SROB}. 
For each rent client $i \in R_j$, we define $i$'s cost share to be
$2^{j+1}$. We now show that the Steiner cost and the assignment cost
of rent clients is at most thrice the total cost share $\sum_j
2^{j+1}|R_j|$.

We will need the following lemma which says that class-$j$ buy clients
are at least $2^{j-1}$-apart from each other.

\begin{lemma}
    \label{lem:CFL:sep}
    $d(z,z') \geq 2^{j-1}$ for $z,z' \in Z_j$.
\end{lemma}

\begin{proof}
  By triangle inequality, we have $d(z,z') \geq d(z,\sigma(z')) -
  d(z', \sigma(z'))$.  Suppose $z$ arrived after $z'$. Since $a_z$ is
  defined to be the distance from $z$ to the nearest open facility
  when $z$ arrived and the facility $\sigma(z')$ that $z'$ was
  assigned to is open at this time, we have $d(z,\sigma(z')) \geq
  a_z$. Moreover, $z'$ is a buy client and so we assigned it to a
  facility $\sigma(z')$ such that $d(z',\sigma(z')) <
  \frac{1}{4}a_{z'}$. We now have that $d(z,z') \geq a_z -
  \frac{1}{4}a_{z'}$. Since $z, z'$ are of class $j$, we have $a_z
  \geq 2^j$ and $a_{z'} < 2^{j+1}$. Thus, $d(z,z') \geq 2^{j-1}$.
\end{proof}

\begin{lemma}
  \label{lem:CFL:buy}
  $Mc(H) + \sum_{i \in R} a_i \leq 3\sum_j 2^{j+1}|R_j|$.
\end{lemma}

\begin{proof}
  Since the cost share of a rent terminal $i \in R$ is at least $a_i$,
  it suffices to show that $Mc(H) \leq 2\sum_{i \in R}a_i$.  We will
  prove this using the following two claims.

\begin{claim}
  $c(H) \leq \sum_{z \in Z} 2a_z$.
\end{claim}

\begin{proof}
  We proceed by charging the increase in $c(H)$ due to an open
  facility to the buy client that opened it. For each open facility $y
  \in F'$, define $z(y) \in Z$ to be the buy client that opened it and
  $x(y) \in F'$ to be the open facility that $y$ was connected to when
  it was opened. So $c(H) = \sum_{y \in F'} d(y,x(y))$. We now show
  that $d(y,x(y)) \leq 2 a_{z(y)}$ for each $x \in F'$.

  Fix an open facility $y \in F'$. For brevity, we write $z$, and $x$,
  in place of $z(y)$, and $x(y)$, respectively. By triangle inequality,
  $d(y,x) \leq d(y,z) + d(z,x)$ so now we bound the right-hand side in
  terms of $a_z$. 
  By definition of the algorithm, at the time when $z$ just arrived,
  before it opened $y$, we have that $x$ was the nearest open facility
  and $y$ was the nearest virtual facility. So $a_z = d(z,x)$ and $y =
  \vassign(z)$. Furthermore, $d(z,\vassign(z)) < \frac{1}{4}a_z$ since
  $z$ is not a virtual client. Thus, $d(y,x) \leq d(y,z) + d(z,x) \leq
  2a_z$, as desired.

  Each open facility was opened by a distinct buy client so 
  $c(H) = \sum_{y \in F'} d(y,x(y)) \leq \sum_{z \in Z} 2a_z$.
\end{proof}

\begin{claim}
  $\sum_{z \in Z} Ma_z \leq \sum_j 2^{j+1}|R_j|$.
\end{claim}

\begin{proof}
  We have $\sum_{z \in Z}Ma_z \leq \sum_j 2^{j+1}M|Z_j|$.  We now
  prove that $M|Z_j| \leq |R_j|$ for each class $j$. The witness set
  $W(z)$ of $z \in Z_j$ satisfies the following properties: (1)
  $|W(z)| \geq M$; (2) $W(z) \subseteq R_j$; (3) $d(i,z) < 2^{j-2}$
  for $i \in W(z)$. The first implies that $M|Z_j| \leq \sum_{z \in
    Z_j} |W(z)|$. Lemma \ref{lem:CFL:sep} together with the second and
  third properties imply that the witness sets of $Z_j$ are disjoint
  subsets of $R_j$ and so $\sum_{z \in Z_j} |W(z)| \leq
  |R_j|$. These two inequalities imply that $\sum_{z \in
    Z}Ma_z \leq \sum_j 2^{j+1}|R_j|$.
\end{proof}

Combining these claims, we have $Mc(H) \leq \sum_{z \in Z}2Ma_z \leq
2\sum_j 2^{j+1}|R_j|$. Therefore, $Mc(H) + \sum_{i \in R} a_i \leq
3\sum_j 2^{j+1}|R_j|$. 
\end{proof}

Now we show that the total cost share is at most a constant times the
cost of the optimal rent-or-buy solution on any HST embedding of the
terminals.

\begin{lemma}
  \label{lem:CFL:sharetree}
  $\sum_j 2^{j+1}|R_j| \leq O(1) \OPTROB(T)$ for any HST embedding $T$
  of $(X,d)$.
\end{lemma}

\begin{proof} 
  Let $T$ be a HST embedding. 
  We charge the cost share of class-$j$ clients to the level-($j-2$)
  cuts. So overall, we have $\sum_j 2^{j+1}|R_j| = \sum_j \sum_{C \in
    \C_j(T)} 2^{j+3} |R_{j+2} \cap C|$. 
  The rest of the proof proceeds as in the proof of Lemma
  \ref{lem:SROB:cost}.
\end{proof}

Now we have all the required ingredients to bound the competitive
ratio of Algorithm \ref{alg:CFL}. By Corollary \ref{cor:main} and Lemma
\ref{lem:CFL:sub}, there exists a HST embedding $T^*$ such that
$\OPTROB(T^*) \leq O(\log k)\OPTROB \leq O(\log k)\OPT$. So Lemmas
\ref{lem:CFL:buy} and \ref{lem:CFL:sharetree} imply that the Steiner
cost and the assignment cost of rent clients is $Mc(H) + \sum_{i \in
  R} a_i \leq O(\log k)\OPT$. Lemma \ref{lem:CFL:FL} says that the
remainder of the algorithm's cost is at most $O(\log k)\OPT$
as well. Thus it is $O(\log k)$-competitive for connected facility
location and this proves Theorem \ref{thm:CFL}.




\section{Prize-Collecting Steiner Tree}
\label{sec:PCST}

In this section, we give a simple algorithm and analysis for the prize-collecting Steiner tree problem and prove the following theorem. 
\begin{theorem}[\cite{Qian:2011wg}]
  \label{thm:PCST}
  There is a deterministic $O(\log k)$-competitive algorithm for the
  online prize-collecting Steiner tree problem.
\end{theorem}
We recall the problem statement. The algorithm is given a root terminal $r$ initially and maintains a subgraph $H$ online. At each online step, a terminal $i$ with \emph{penalty} $\pi_i$ arrive and the algorithm can either pay the penalty or augment $H$ such that $H$ connects $i$ to the root. We say that $i$ is a \emph{penalty terminal} if the algorithm chose to pay $i$'s penalty and denote by $P$ the set of penalty terminals. The total cost of the algorithm is $c(H) + \sum_{i \in P} \pi_i$.

\paragraph{Algorithm.}
Algorithm \ref{alg:prize} maintains a set of \emph{buy terminals} $Z$ ($Z$ includes $r$) which it connects to the root. It also associates to each terminal $i$ a cost share $\rho_i$. When a terminal $i$ with penalty $\pi_i$ arrives, let $z \in Z$ be the closest buy terminal. We define $a_i = d(i,z)$, and define $\class(i) = j$ if $a_i \in [2^j, 2^{j+1})$. The algorithm initializes $i$'s cost share $\rho_i$ to $0$ and raises $\rho_i$ until either the total cost share of class-$j$ terminals within a radius of $2^{j-1}$ of $i$ (including $i$) is at least $2^{j+1}$ (sufficient to pay for the edge $(i,z)$) or $\rho_i = \pi_i$. In the first case, we add the edge $(i,z)$ to $H$; in the second case, we pay the penalty. We denote by $X_j$ the set of class-$j$ terminals and $W(i)$ the set of class-$j$ terminals within a radius of $2^{j-1}$ of $i$ (we call $W(i)$ the \emph{witness set} of $i$).

\begin{algorithm}
  \caption{Algorithm for Online Prize-Collecting Steiner Tree}
  \begin{algorithmic}[1]
    \label{alg:prize}
    \STATE $Z \leftarrow \{r\}; H \leftarrow \emptyset$; $X_j \leftarrow \emptyset$ for all $j$
    \WHILE {terminal $i$ with penalty $\pi_i$ arrives}
    \STATE Let $z$ be closest terminal in $Z$ to $i$, 
    set $j = \lfloor \log d(i,z) \rfloor$ and add $i$ to $X_j$
    \STATE $W(i) \leftarrow \{i' \in X_j : d(i,i') < 2^{j-1}\}$
    \STATE Initialize $\rho_i \leftarrow 0$ and increase $\rho_i$ until $\sum_{i' \in W(i)} \rho_{i'} \geq 2^{j+1}$ or $\rho_i = \pi_i$
    \IF {$\sum_{i' \in W(i)} \rho_{i'} \geq 2^{j+1}$}
    \STATE Add $i$ to $Z$ and buy $(i,z)$, i.e. $H \leftarrow H \cup \{(i,z)\}$
    \ELSE
    \STATE Pay penalty $\pi_i$
    \ENDIF
    \ENDWHILE
  \end{algorithmic}
\end{algorithm}




\paragraph{Analysis.}
Let $X$ be the set of terminals. For each buy terminal $z \in Z$, the algorithm incurs a cost of $a_z$. Thus, the total cost of the algorithm is $\sum_{z \in Z} a_z + \sum_{i \in P} \pi_i$. Our goal in the analysis is to show that this is at most $\OPT(T)$ for any HST embedding of $(X,d)$.

First, we show that the cost of the algorithm is at most twice the total cost share $\sum_i \rho_i$.  \begin{lemma}
  \label{lem:prize:cost}
  We have that $\sum_{z \in Z} a_z + \sum_{i \in P} \pi_i \leq 2\sum_i \rho_i$.
\end{lemma}

\begin{proof}
  The algorithm pays the penalty $\pi_i$ for terminal $i$ only if $\rho_i = \pi_i$ so $\sum_{i \in P} \pi_i \leq \sum_i \rho_i$. Let $Z_j \subseteq Z$ be the set of class-$j$ buy terminals. Since $a_z \in [2^j, 2^{j+1})$ for $z \in Z_j$, we have $\sum_{z \in Z} a_z \leq \sum_j 2^{j+1} |Z_j|$.  We now show that $2^{j+1}|Z_j| \leq
  \sum_{i \in X_j} \rho_i$ for each class $j$.
  
  Fix a class $j$. The witness set $W(z)$ of $z \in Z_j$ satisfies the
  following properties: (1) $\sum_{i \in W(z)} \rho_i \geq 2^{j+1}$; (2) $W(z) \subseteq X_j$; (3)
  $d(i,z) < 2^{j-1}$ for $i \in W(z)$. The first implies that
  $2^{j+1}|Z_j| \leq \sum_{z \in Z_j} \sum_{i \in W(z)} \rho_i$. We claim that $d(z,z') \geq
  2^j$ for $z,z'\in Z_j$. This completes the proof since together with
  the second and third properties, we have that the witness sets of
  $Z_j$ are disjoint subsets of $X_j$ and so $\sum_{z \in Z_j} \sum_{i \in W(z)} \rho_i
  \leq \sum_{i \in X_j} \rho_i$.

  Now we prove the claim. Observe that $H$ is the
  subgraph produced by the online greedy Steiner tree algorithm if it
  were run on the subsequence of buy terminals $Z$, and for each $z \in
  Z$ we have that $a_z$ is exactly the distance from $z$ to the
  nearest previously-arrived buy terminal. Thus, we can apply Lemma
  \ref{lem:r-sep} and get that $d(z,z') \geq 2^j$ for any $z,z' \in
  Z_j$, proving the claim. Putting all of the above together, we get
  $\sum_{z \in Z} a_z \leq \sum_j \sum_{i \in X_j} \rho_i = \sum_i \rho_i$. This finishes the
  proof of the lemma.
\end{proof}


Next, we show that the total cost share is at most a constant times the cost of the optimal solution on any HST embedding.

\begin{lemma}
  \label{lem:prize:C}
  $\sum_i \rho_i \leq 8\OPT(T)$ for any HST embedding $T$ of $(X,d)$.
\end{lemma}

\begin{proof}
  Let $T$ be a HST embedding of $(X,d)$. First, we lower bound
  $\OPT(T)$ in terms of $T$'s cuts. Define $R_j \subseteq X_j$ to be the subset of class-$j$ terminals $i$ with $\rho_i > 0$. Consider a level-$j$ cut $C_e \in \C_j(T)$ and associate the terminals $R_{j+1} \cap C_e$ to $C_e$. (Note that a terminal can only be associated to one cut.) By
  definition, $e \in E_j(T)$ and has length $2^{j-1}$. If $r \notin C_e$ then $e$ lies on the $(i,r)$ path in
  $T$ for each associated terminal $i \in R_{j+1} \cap C_e$ so the optimal solution on $T$ either pays the penalty for
  each associated terminal at a cost of $\sum_{i \in R_{j+1} \cap C_e} \pi_i$ or buys $e$ at a cost of $2^{j-1}$. Since each terminal is associated to a unique cut, we have \[\OPT(T) \geq
  \sum_j \sum_{C \in C_j(T) : r \notin C} \min\left\{\sum_{i \in R_{j+1} \cap C}
  \pi_i, 2^{j-1}\right\}.\]

We have $\sum_i \rho_i = \sum_j \sum_{i \in R_j} \rho_i$. For each terminal $i \in R_j$, we charge $\rho_i$ to the level-($j-1$) cut $C \in \C_{j-1}$ containing $i$. Each level-$j$ cut $C \in \C_j(T)$ is charged $\sum_{i \in R_{j+1} \cap C} \rho_i$. Thus, we have \[\sum_i \rho_i = \sum_j \sum_{C \in \C_j(T)} \left(\sum_{i \in R_{j+1} \cap C} \rho_i\right).\]

Since $\rho_i \leq \pi_i$ for each terminal $i$, it suffices to prove the following claim: for each level-$j$ cut $C \in \C_j(T)$, we have $\sum_{i \in R_{j+1} \cap C} \rho_i = 0$ if $r \in C$, and $\sum_{i \in R_{j+1} \cap C} \rho_i \leq 2^{j+2}$ if $r \notin C$. Suppose $r \in C$. Since $a_i < d(i,r) < 2^j$ for each $i \in C$, there cannot be any class-($j+1$) terminal in $C$ and so $\sum_{i \in R_{j+1} \cap C} \rho_i = 0$. Now consider the case $r \notin C$. Suppose, towards a contradiction, that $\sum_{i \in R_{j+1} \cap C} \rho_i > 2^{j+2}$. Let $i^*$ be the last-arriving terminal of $R_{j+1} \cap C$, i.e. $i^*$ is the last-arriving class-$(j+1)$ terminal of $C$ with $\rho_{i^*} > 0$. Since the diameter of $C$ is less than $2^j$, we have $W(i^*) \supseteq R_{j+1} \cap C$.  We also have $\sum_{i \in R_{j+1} \cap C \setminus \{i^*\}} \rho_{i} < 2^{j+2}$ because otherwise the algorithm would not have increased $\rho_{i^*}$. But the algorithm increased $\rho_{i^*}$ ensuring that $\sum_{i \in W(i^*)} \rho_{i} \leq 2^{j+1}$, contradicting the assumption that $\sum_{i \in R_{j+1} \cap C} \rho_i > 2^{j+1}$. This completes the proof of the claim and so we get $\sum_i \rho_i \leq 8\OPT(T)$.  \end{proof}

Now, Lemmas \ref{lem:prize:cost} and \ref{lem:prize:C} imply that the cost of Algorithm \ref{alg:prize} is at most $O(1)\OPT(T)$ for any HST embedding $T$ of $(X,d)$. Furthermore, by Theorem \ref{thm:FRT}, there exists a HST embedding $T^*$ such that $\OPT(T^*) \leq O(\log k)\OPT$. 
Thus, the algorithm is $O(\log k)$-competitive for prize-collecting Steiner tree and this proves Theorem \ref{thm:PCST}.


\section*{Acknowledgements} The author is very grateful to Shuchi
Chawla and Mohit Singh for helpful discussions and to the anonymous referees for their useful comments.




\bibliographystyle{plain}
\bibliography{bibliography}

\appendix


\end{document}